\documentclass{article}

\usepackage{arxiv}
\usepackage{amssymb}
\usepackage{amsthm}
\usepackage{enumerate}
\usepackage{amsmath,bm}
\usepackage{hyperref}
\usepackage[section]{placeins}
\usepackage[utf8]{inputenc} 
\usepackage[T1]{fontenc}    
\usepackage{hyperref}       
\usepackage{url}            
\usepackage{booktabs}       
\usepackage{amsfonts}       
\usepackage{nicefrac}       
\usepackage{microtype}      
\usepackage{lipsum}
\usepackage{graphicx}
\newtheorem{theorem}{Theorem}[section]

\newtheorem{lemma}[theorem]{Lemma}

\graphicspath{ {./images/} }

\title{New record-breaking binary linear codes constructed from group codes}

\author{
 Cong Yu \\
  School of Mathematics\\
  Hefei University of Technology\\
  Hefei, Anhui, People’s Republic of China, 230009  \\
  \texttt{zyu6952@gmail.com} \\
   \And
 Shixin Zhu \\
  School of Mathematics\\
  Hefei University of Technology\\
  Hefei, Anhui, People’s Republic of China, 230009 \\
  \texttt{zhushixinmath@hfut.edu.cn} \\
  \And
  Hao Chen \\
  College of Information Science and Technology\\
  Jinan University\\
  Guangzhou, Guangdong, People’s Republic of China, 510632 \\
  \texttt{haochen@jnu.edu.cn} \\
  \And
 Yang Li \\
  School of Mathematics\\
  Hefei University of Technology\\
  Hefei, Anhui, People’s Republic of China, 230009 \\
  \texttt{yanglimath@163.com} \\
  \And
 Xiuyu Zhang\\
  School of Mathematics\\
  Hefei University of Technology\\
  Hefei, Anhui, People’s Republic of China, 230009 \\
  \texttt{xyzfuman@mail.hfut.edu.cn} \\
}

\begin{document}
\maketitle
\begin{abstract}
In this paper, we employ group rings and automorphism groups of binary linear codes to construct new record-breaking binary linear codes. We consider the semidirect product of abelian groups and cyclic groups and use these groups to construct linear codes. Finally, we obtain some linear codes which have better parameters than the code in \cite{bib5}. All the calculation results and corresponding data are listed in the paper or posted online.
\end{abstract}


\section{Introduction}
Let $\mathbb{F}_q $ be the finite field with $q$ elements. A $q$-ary $[n,k,d]$ linear code $\mathcal{C}$ is a $k$-dimensional subspace of $\mathbb{F}_q^n$, where $d$ is the minimum distance of $\mathcal{C}$. For $\mathbf{a}=(a_1,a_2,\dots ,a_n), \mathbf{b}=(b_1,b_2,\dots ,b_n)\in \mathbb{F}_q^n$, the ordinary inner-product is defined by $[\mathbf{a},\mathbf{b}]=\sum_{i=1}^na_ib_i$. The dual code of $\mathcal{C}$ is defined by
$$\mathcal{C}^{\bot}=\left\lbrace \mathbf{c}\in \mathbb{F}_q^n \mid [\mathbf{x},\mathbf{c}]=0, \forall \mathbf{x} \in \mathcal{C} \right\rbrace. $$ Two  linear codes $\mathcal{C}_1 $ and $\mathcal{C}_2$ are  permutation equivalent if there is a permutation $\pi \in S_n$ such that $\mathcal{C}_1=\pi(\mathcal{C}_2)$, where $S_n$ is the symmetric group on $n$ letters and $\pi(\mathcal{C})=\left\lbrace \pi(\mathbf{c}) \mid \mathbf{c} \in \mathcal{C} \right\rbrace$. The permutation automorphism group of $\mathcal{C}$ is denoted by $\text{PAut}(\mathcal{C})=\left\lbrace \pi \in S_n \mid \pi(\mathcal{C})=\mathcal{C} \right\rbrace $.
\begin{lemma}\cite{bib3}
	Let $\mathcal{C}$ be a binary linear code. Then $\mathrm{PAut}(\mathcal{C})=\mathrm{PAut}(\mathcal{C}^{\bot})$.
\end{lemma}
For fixed length $n$ and dimension $k$, finding the largest possible value of minimum distance $d$ is an open problem in coding theory. In \cite{bib5}, Grassl provides online code tables  that contains linear codes with the best known parameters. However, there are gaps between the lower and upper bounds on the minimum distance in the code tables. Improving the lower bounds or narrowing the upper bounds is a meaningful work in coding theory. In recent years, many scholars have made great contribution on improving the lower bound on the minimum distance in the code tables. Since cyclic codes, quasi-cyclic codes and constacyclic codes  have good algebraic structure, they were used to construct new linear codes, for details see \cite{bib7},\cite{bib13},\cite{bib12},\cite{bib8},\cite{bib10},\cite{bib9},\cite{bib11}.
\\ \indent Group rings are useful tools in constructing codes. In \cite{bib1}, Hurley  proposed group ring matrix(or $G$-matrix). Group codes are ideals in group rings. In \cite{bib2}, the authors  used  $G$-matrix to study group codes. Since then, $G$-matrix have been widely used to construct some interesting codes(see \cite{bib18},\cite{bib22},\cite{bib23},\cite{bib19},\cite{bib21},\cite{bib14},\cite{bib20}). The code generated by  $G$-matrix is an interesting class of codes because the group $G$ is isomorphic to a subgroup of the automorphism group of the code. Linear codes with large automorphism groups can reduce the amount of computation for encoding and decoding(see \cite{bib16},\cite{bib17}). However, in \cite{bib5}, the order of automorphism groups of most linear codes is very small, generally not exceeding their code lengths, and some are even trivial. This is not practical in the process of encoding and decoding. The order of the automorphism group of a code generated by a $G$-matrix is at least equal to the code length. Therefore, when constructing a linear code with a $G$-matrix, we can obtain good parameters and at the same time ensure that the order of the automorphism group is at least $n$, where $n$ is the code length. In \cite{bib4}, Huffman gave a method for decomposing a linear code by using an automorphism of the code. Using this idea we can decompose a linear code into a direct sum of two subcodes through an automorphism of the code. Furthermore, if the properties of the two subcodes are better, a new record-breaking code can be constructed by the well-known Construction X.
\\ \indent The paper is structured as follows: In Sec. \ref{Sec 2}, we give some basic concepts about group rings and $G$-matrix. In Sec. \ref{Sec.3}, we study the $G$-matrix of several classes of groups which are semidirect products of abelian groups and cyclic groups. In Sec. \ref{Sec.4}, we give the idea of using the automorphism of codes to construct new record-breaking linear codes. In Sec. \ref{Sec 5}, we give the final calculation results and list them in the table. In Sec. \ref{Sec 6}, we give some possible directions for future research.

\section{Group ring and $G$-matrix}\label{Sec 2}
Let $G=\left\lbrace g_1,g_2,\dots,g_n\right\rbrace $ be a finite group of order $n$. A  group code is an ideal of a group ring $\mathbb{F}_qG$, where $\mathbb{F}_qG=\left\lbrace \sum_{i=1}^n\alpha_{g_i}g_i \mid \alpha_i \in \mathbb{F}_q \right\rbrace $. The addition in $\mathbb{F}_qG$ is defined by $$
\sum_{i=1}^n\alpha_{g_i}g_i+\sum_{i=1}^n\beta_{g_i}g_i=\sum_{i=1}^n(\alpha_{g_i}+\beta_{g_i})g_i,
$$ and the product of two elements in $\mathbb{F}_qG$ is given by:
$$\left(\sum_{i=1}^{n} \alpha_{g_{i}} g_{i}\right)\left(\sum_{j=1}^{n} \beta_{g_{j}} g_{j}\right)=\sum_{i=1}^{n}
\left (   \sum_{j=1}^{n}\alpha_{g_{i}g_{j}^{-1}} \beta_{g_{j}}\right )g_{i}.$$ The group ring matrix was first introduced in \cite{bib1}. Let $v=\sum_{i=1}^{n} \alpha_{g_{i}} g_{i}\in \mathbb{F}_qG$, there is a natural bijection $\Psi$ from $\mathbb{F}_qG$ to $\mathbb{F}_q^n$: $$\Psi(v)=(\alpha_{g_1},\alpha_{g_2},\dots ,\alpha_{g_n}).$$ Thus, in the later content, we can regard an elements in $\mathbb{F}_qG$ as a vector in $\mathbb{F}_q^n$.  Define the map $\sigma_G$ from $\mathbb{F}_qG$ to $M_{n}\left( \mathbb{F}_q\right) $ to be 
$$\sigma_G(v) =\left(\begin{array}{ccccc}
	\alpha_{g_{1}^{-1} g_{1}} & \alpha_{g_{1}^{-1} g_{2}} & \alpha_{g_{1}^{-1} g_{3}} & \cdots & \alpha_{g_{1}^{-1} g_{n}} \\
	\alpha_{g_{2}^{-1} g_{1}} & \alpha_{g_{2}^{-1} g_{2}} & \alpha_{g_{2}^{-1} g_{3}} & \cdots & \alpha_{g_{2}^{-1} g_{n}} \\
	\vdots & \vdots & \vdots & \vdots & \vdots \\
	\alpha_{g_{n}^{-1} g_{1}} & \alpha_{g_{n}^{-1} g_{2}} & \alpha_{g_{n}^{-1} g_{3}} & \cdots & \alpha_{g_{n}^{-1} g_{n}}
\end{array}\right),$$ where $M_n(\mathbb{F}_q)$ is the set of all $n\times n$ matrix over $\mathbb{F}_q$. The matrix $\sigma_G(v)$ is called a group ring matrix or $G$-matrix. The code generated by the row space of $\sigma_G(v)$ is denoted by $C(v)$. It is not difficult to verify that $C(v)$ is a left principal ideal of $\mathbb{F}_qG$.
\begin{lemma}(\cite{bib1})
	$\sigma_G$ is an injective ring homomorphism, i.e. for $v_1,v_2\in \mathbb{F}_qG$, $\sigma_G(v_1v_2)=\sigma_G(v_1)\sigma_G(v_2)$.
\end{lemma}
For an element $v=\sum_{i=1}^{n} \alpha_{g_{i}} g_{i}\in \mathbb{F}_qG$, we denote $v^{(s)}=\sum_{i=1}^{n} \alpha_{g_{i}} g_{i}^s$. Then 
\begin{lemma}(\cite{bib2}) \label{Lemmma 2.3}
	Let $G$ be a finite group of order $n$. For an element $v=\sum_{i=1}^{n} \alpha_{g_{i}} g_{i}\in \mathbb{F}_qG$, we have that $\sigma_G(v^{(-1)})=\sigma_G(v)^T$, where $\sigma_G(v)^T$ is the transpose of $\sigma_G(v)$. Moreover, the group $G$ isomorphic to a subgroup of $\mathrm{PAut}(C(v))$.	
\end{lemma}
\section{Several classes of groups and their $G$-matrix}\label{Sec.3}
Let us review the semidirect product of two groups. Semidirect products are direct products in which the componentwise operation is twisted by a group-on-group action. Given two groups $A$ and $B$, let $\varphi$ be a homomorphism from $B$ to $\text{Aut}(A)$, where $\text{Aut}(A)$ is the automorphism group of $A$. The semidirect product $A\rtimes_{\varphi}B$ with the multiplication is defined by $$
(a_1,b_1)(a_2,b_2)=(a_1\varphi(b_1)(a_2),b_1b_2)
$$ for all $a_1,a_2 \in A$ and $ b_1,b_2 \in B$. Notice that when $\varphi$ is a trivial homomorphism, then $A\rtimes_{\varphi}B$ is the direct product $A\times B$ with
componentwise multiplication.
\\ \indent Let $\mathbf{a}=(a_1,a_2,\dots ,a_n)\in \mathbb{F}_q^n$, we call the matrix
$$A=\begin{pmatrix}
	a_1 & a_2 &\dots   & a_n\\
	a_{n-1}&a_1  & \dots  & a_{n-2}\\
	\vdots& \vdots &\dots  & \vdots\\
	a_2& a_3 &\dots  &a_1
\end{pmatrix}$$ the circulant matrix generated by $\mathbf{a}$, denoted by $A=circ(\mathbf{a})$. Similarly, we denote the block circulant matrix by $Circ(A_1,A_2,\dots,A_n)$, where $A_i$ are matrices of the same type, $1\le i \le n$. We call an $n$ by $n$ matrix  a $\lambda$-circulant matrix if each element of its row vector is the result of moving each element of the previous row vector $\lambda$  position to the right, and we denote it by $circ_{\lambda}(a_1,a_2,\dots ,a_n).$ For example, $$
circ_{2}(a_1,a_2,a_3,a_4)=\begin{pmatrix}
	a_1 & a_2 & a_3 &a_4 \\
	a_3 & a_4 & a_1 & a_2\\
	a_1 & a_2 & a_3 & a_4\\
	a_3&a_4  & a_1 &a_2
\end{pmatrix}.
$$
 Let $[a]_n$ denote the smallest positive integer $a'$ such that $a\equiv a' \mod n$. Notice that when $\lambda=1$, the 1-circulant matrix is actually a circulant matrix and if $[\lambda_1]_n= [\lambda_2]_n$, then $circ_{\lambda_1}({\bf a})=circ_{\lambda_2}({\bf a})$. We define $T_n$ as a cyclic shift operator on a vector of length $n$, i.e. $T_n(b_1,b_2,\dots ,b_n)=(b_n,b_1,\dots ,b_{n-1})$, where $b_i$ can be either a vector or an element. Let $C_n$ be the cyclic group of order $n$. Next we study $G$-matrix of several classes of groups.
\subsection{The group $C_n\rtimes_{\varphi}C_m$}
First we consider the semidirect product of two cyclic groups. Let $$G_1=C_n\rtimes_{\varphi}C_m=\left\langle y \right\rangle \rtimes_{\varphi}\left\langle x \right\rangle =\left\langle x,y\mid x^m=y^n=1, y^x=y^k \right\rangle, $$ where $k$ satisfies $k^m
\equiv 1\mod n$ and $k^t\not\equiv 1 \mod n$ for $1\le t\le m-1$. Notice that  $k=1$ indicates $\varphi$ is trivial and $m=2, k=n-1$ indicates that $G_1$ is actually a dihedral group $D_n$. Let $v=\sum_{j=0}^{m-1}\sum_{i=0}^{n-1}a_{1+i+nj}x^jy^i\in \mathbb{F}_qG_1$, we have the following theorem:
 \begin{theorem}\label{th 3.1}
 	If $(b_1,b_2,\dots ,b_{mn})$ is a codeword in $C(v)$, so are $(\mathbf{b}_1,\mathbf{b}_2,\dots,\mathbf{b}_m)$ and $T_m(\mathbf{b}_1',\mathbf{b}_2',\dots,\mathbf{b}_{m}')$, where $\mathbf{b}_{j+1}=T^{[k^j]_n}_n(b_{1+nj},b_{2+nj}, \dots ,b_{n+nj})$ and $\mathbf{b}_{j+1}'=(b_{1+nj},b_{2+nj},\dots ,b_{n+nj})$, $0\le j \le m-1$.
 \end{theorem}
 \begin{proof}
 	Let $v'=\sum_{j=0}^{m-1}\sum_{i=0}^{n-1}b_{1+i+nj}x^jy^i\in \mathbb{F}_qG_1$. Then 
 	$$
 	\begin{aligned}
 		xv'	&=\sum_{j=0}^{m-1}\sum_{i=0}^{n-1}b_{1+i+nj}x^{j+1}y^i\\ &=\sum_{i=0}^{n-1}b_{1+i+n(m-1)}y^i+\sum_{j=0}^{m-2}\sum_{i=0}^{n-1}b_{1+i+nj}x^{j+1}y^i.
 	\end{aligned}
 	$$
 	Thus $\Psi(xv')=(\mathbf{b}_m',\mathbf{b}_1',\dots,\mathbf{b}_{m-1}')\in C(v)$. 
 	$$
 	\begin{aligned}
 		yv'&=\sum_{j=0}^{m-1}\sum_{i=0}^{n-1}b_{1+i+nj}yx^{j}y^i\\ &=\sum_{j=0}^{m-1}\sum_{i=0}^{n-1}b_{1+i+nj}xy^kx^{j-1}y^i\\
 		&=\sum_{j=0}^{m-1}\sum_{i=0}^{n-1}b_{1+i+nj}x^2y^{k^2}x^{j-2}y^i\\& \dots 
 		\\ &=\sum_{j=0}^{m-1}\sum_{i=0}^{n-1}b_{1+i+nj}x^jy^{[k^j]_n+i}.	
 	\end{aligned}
 	$$
 	Thus $\Psi(yv')=(\mathbf{b}_1,\mathbf{b}_2,\dots,\mathbf{b}_m)\in C(v)$.
 \end{proof}
 Then from Theorem \ref{th 3.1}, we can immediately get
 \begin{equation}\label{eq 1}
 	\sigma_{G_1}(v)=Circ(A_1,A_2,\dots ,A_m),
 \end{equation} where $$A_{j+1}=circ_{[k^{j}]_n}(a_{1+nj},a_{2+nj},\dots ,a_{n+nj}), 0\le j \le m-1.$$ 
 \subsection{The group $(C_{n_1}\times C_{n_2})\rtimes_{\varphi} C_m$}
 Next we consider the semidirect product of the direct product of two cyclic groups and a cyclic group. Let $$G_2=(C_{n_1}\times C_{n_2})\rtimes_{\varphi} C_m=\left\langle x,y,z\mid x^{n_1}=y^{n_2}=z^m=1,x^z=x^{k_1},y^z=y^{k_2} \right\rangle, $$ where $k_1,k_2$ satisfy $k_i^m
 \equiv 1\mod n_i$ and $k_i^{t}\not\equiv 1 \mod n_i$ for $1\le t\le m-1$, $i=1,2.$ Let $v=\sum_{i=0}^{m-1}\sum_{j=0}^{n_2-1}\sum_{k=0}^{n_1-1}a_{1+k+n_1j+n_1n_2i}z^iy^jx^k\in \mathbb{F}_qG_2$, we have the following theorem:
 \begin{theorem}\label{th 3.2}
 	If $(b_1,b_2,\dots ,b_{n_1n_2m})$ is a codeword in $C(v)$, so are $(\widetilde{\mathbf{b}_1},\widetilde{\mathbf{b}_2},\dots, \widetilde{\mathbf{b}_m} )$, $(\widetilde{\mathbf{b}_1'},\widetilde{\mathbf{b}_2'},\dots, \widetilde{\mathbf{b}_m'} )$ and $T_m(\widetilde{\mathbf{b}_1''},\widetilde{\mathbf{b}_2''},\dots, \widetilde{\mathbf{b}_m''} )$ where 
 	$$
 	\widetilde{\mathbf{b}_{i+1}}=(\mathbf{b}_1^{(i+1)},\mathbf{b}_2^{(i+1)},\dots ,\mathbf{b}_{n_2}^{(i+1)}),
 	$$
 	$$
 	\widetilde{\mathbf{b}_{i+1}'}=T^{[k_2^i]_{n_2}}_{n_2}(\mathbf{b}_1'^{(i+1)},\mathbf{b}_2'^{(i+1)},\dots , \mathbf{b}_{n_2}'^{(i+1)}),
 	$$
 	$$
 	\widetilde{\mathbf{b}_{i+1}''}=(b_{1+n_1n_2i},b_{2+n_1n_2i},\dots ,b_{n_1n_2(i+1)}),
 	$$
 	$$
 	\mathbf{b}_{j+1}^{(i+1)}=T^{[k_1^i]_{n_1}}_{n_1}(b_{in_1n_2+jn_1+1},b_{in_1n_2+jn_1+2},\dots ,b_{in_1n_2+(j+1)n_1}),
 	$$
 	$$
 	\mathbf{b}_{j+1}'^{(i+1)}=(b_{in_1n_2+jn_1+1},b_{in_1n_2+jn_1+2},\dots ,b_{in_1n_2+(j+1)n_1}),
 	$$
 	$$0\le i \le m-1, 0\le j \le n_2-1.$$
 \end{theorem}
 \begin{proof}
 	Let $v'=\sum_{i=0}^{m-1}\sum_{j=0}^{n_2-1}\sum_{k=0}^{n_1-1}b_{1+k+n_1j+n_1n_2i}z^iy^jx^k\in C(v)$. Then $$
 	\begin{aligned}
 		xv'	&=x\sum_{i=0}^{m-1}\sum_{j=0}^{n_2-1}\sum_{k=0}^{n_1-1}b_{1+k+n_1j+n_1n_2i}z^iy^jx^k\\
 		&=\sum_{i=0}^{m-1}\sum_{j=0}^{n_2-1}\sum_{k=0}^{n_1-1}b_{1+k+n_1j+n_1n_2i}xz^iy^jx^k\\
 		&=\sum_{i=0}^{m-1}\sum_{j=0}^{n_2-1}\sum_{k=0}^{n_1-1}b_{1+k+n_1j+n_1n_2i}zx^{k_1}z^{i-1}y^jx^k\\
 		&=\dots \\
 		&=\sum_{i=0}^{m-1}\sum_{j=0}^{n_2-1}\sum_{k=0}^{n_1-1}b_{1+k+n_1j+n_1n_2i}z^ix^{[k_1^i]_{n_1}}y^jx^k\\
 		&=\sum_{i=0}^{m-1}\sum_{j=0}^{n_2-1}\sum_{k=0}^{n_1-1}b_{1+k+n_1j+n_1n_2i}z^iy^jx^{k+[k_1^i]_{n_1}}.
 	\end{aligned}
 	$$
 	Thus $\Psi(xv')=(\widetilde{\mathbf{b}_1},\widetilde{\mathbf{b}_2},\dots, \widetilde{\mathbf{b}_m} )\in C(v)$.
 	 $$
 	\begin{aligned}
 		yv'	&=y\sum_{i=0}^{m-1}\sum_{j=0}^{n_2-1}\sum_{k=0}^{n_1-1}b_{1+k+n_1j+n_1n_2i}z^iy^jx^k\\
 	  &=\sum_{i=0}^{m-1}\sum_{j=0}^{n_2-1}\sum_{k=0}^{n_1-1}b_{1+k+n_1j+n_1n_2i}yz^iy^jx^k\\
 	  &=\sum_{i=0}^{m-1}\sum_{j=0}^{n_2-1}\sum_{k=0}^{n_1-1}b_{1+k+n_1j+n_1n_2i}z^iy^{j+[k_2^i]_{n_2}}x^k.
 	\end{aligned}
 	$$
 	Thus $\Psi(yv')=(\widetilde{\mathbf{b}_1'},\widetilde{\mathbf{b}_2'},\dots, \widetilde{\mathbf{b}_m'} )\in C(v)$.
 	 $$
 	\begin{aligned}
 		zv'	&=z\sum_{i=0}^{m-1}\sum_{j=0}^{n_2-1}\sum_{k=0}^{n_1-1}b_{1+k+n_1j+n_1n_2i}z^iy^jx^k\\
 		&=\sum_{i=0}^{m-1}\sum_{j=0}^{n_2-1}\sum_{k=0}^{n_1-1}b_{1+k+n_1j+n_1n_2i}z^{i+1}y^jx^k. 
 	\end{aligned}
 	$$ Thus $\Psi(zv')=T_m(\widetilde{\mathbf{b}_1''},\widetilde{\mathbf{b}_2''},\dots, \widetilde{\mathbf{b}_m''} )\in C(v)$.
 \end{proof}
 Then from Theorem \ref{th 3.2}, we can immediately get
 \begin{equation}\label{eq 1}
 	\sigma_{G_2}(v)=Circ(A_1,A_2,\dots ,A_m),
 \end{equation} where 
 $$
 A_{i+1}=Circ_{[k_2^i]_{n_2}}(A^{(i+1)}_1,A^{(i+1)}_2,\dots , A^{(i+1)}_{n_2}),
 $$
 $$
 A^{(i+1)}_{j+1}=circ_{[k_1^j]_{n_1}}(a_{1+n_1j+n_1n_2i},a_{2+n_1j+n_1n_2i},\dots ,a_{n_1+n_1j+n_1n_2i}),
 $$
 $$0\le i \le m-1, 0\le j \le n_2-1.$$
 In fact, we can generalize $G_1 ,G_2$ to the following form: $$G_s=(\prod_{t=1}^{s}C_{n_t})\rtimes_{\varphi}C_m.  $$
 $\sigma_{G_s}(v)$ can be obtained by induction for $v\in \mathbb{F}_qG_s$.
 \section{The automorphism groups of $C(v)$}\label{Sec.4}
 In this section, we focus on the binary case. The automorphism group of a binary linear code is actually its permutation automorphism group. A permutation $\pi \in S_n$ is called type $p$-$(c,f)$ if it has $c$ cycles of length $p$ and $f$ fixed points in its decomposition. For example, let $\pi=(1,2,3)(4,5,6)\in S_{13}$, then $\pi$ is of type $3$-$(2,7)$.\\ \indent
 Let $\mathcal{C}$ be a binary linear code of length $n$. Assume that an automorphism $\pi=\Omega_1\Omega_2\cdots \Omega_c\Omega_{c+1}\cdots\Omega_{c+f}\in \text{PAut}(\mathcal{C})$ is of type $p$-$(c,f)$, where $\Omega_1, \dots , \Omega_c$ are disjoint cycles of length $p$ and $\Omega_{c+1},\cdots,\Omega_{c+f}$ are the fixed points. Let $$F_{\pi}(\mathcal{C})=\left\lbrace \mathbf{c}\in \mathcal{C}\mid \pi(\mathbf{c})=\mathbf{c} \right\rbrace $$ and $$E_{\pi}(\mathcal{C})=\left\lbrace \mathbf{c}=(c_1,c_2,\dots,c_n)\in \mathcal{C} \mid \sum_{i\in \Omega_j}c_i \equiv 0 \mod 2 , 1\le j \le c+f \right\rbrace. \ $$ 
 \begin{theorem}(\cite{bib4})\label{th 4.1}
 	With the notations above, if $\gcd(p,2)=1$, then $\mathcal{C}=F_{\pi}(\mathcal{C})\oplus E_{\pi}(\mathcal{C})$. Moreover, $F_{\pi}(\mathcal{C})$ and $E_{\pi}(\mathcal{C})$ are $\pi$-invariant.
 \end{theorem}
 From Theorem \ref{th 4.1} we know that for a given automorphism of a binary linear code, we can decompose this  code into the direct sum form of two subcodes. Our main ideas for constructing linear codes with good parameters are as follows:
 \begin{enumerate}
 	\item Find a code $\mathcal{C}$ whose parameters are the same as or better than the best-known code in \cite{bib5} with an automorphism $\pi$.
 	\item Decompose $\mathcal{C}$ into the direct sum form of two subcodes by Theorem \ref{th 4.1}.
 	\item If $\mathcal{C}$ has parameters $[n,k,d]$ and $E_{\pi}(\mathcal{C})$ has parameters $[n,k-k',d'> d]$. Let $\mathcal{C}'$ be a code with parameters $[n',k',d-d']$, then by using Construction X, we can obtain a code with parameters $[n+n',k,d']$.
 \end{enumerate}
 The key to the above steps is to find $E_{\pi}(\mathcal{C})$ with the largest possible dimension and the minimum distance larger than the original code $\mathcal{C}$. The automorphism groups of the best known linear codes in \cite{bib5} most are either trivial or of very small order, and therefore do not apply to our ideas. From Lemma \ref{Lemmma 2.3} we know that the order of $\text{PAut}(C(v))$ must be a multiple of $\lvert G \rvert$ for $v\in \mathbb{F}_qG$. Thus, we can try to decompose group codes to get some good subcodes.
 \begin{theorem}
 	Let $G=\left\lbrace g_1,g_2,\dots ,g_n \right\rbrace $ be a finite group of order $n$ and $v\in \mathbb{F}_qG$. If $\lvert \mathrm{PAut}(C(v)) \rvert=n$ and  $\pi \in \mathrm{PAut}(C(v))$ is of type $p$-$(c,f), p\ne 1$ then $f=0$. In other words, $\pi$ has no fixed points, where $p$ is the order of $\pi$ and $n=pc$.
 \end{theorem}
 \begin{proof}
 	$\lvert \text{PAut}(\mathcal{C}(v))\rvert=n$ implies that $\text{PAut}(\mathcal{C}(v))\cong G$. Let $v=\sum_{j=1}^n\alpha_{g_j}g_j\in \mathbb{F}_qG$. Then $\text{PAut}(\mathcal{C}(v))=\left\lbrace \pi_{g_1},\dots ,\pi_{g_n} \right\rbrace $, where $\pi_{g_i}(\Psi(v))=(\alpha_{g_i^{-1}g_1},\alpha_{g_i^{-1}g_2},\dots ,\alpha_{g_i^{-1}g_n})$. If $\pi_{g_i}$ has a fixed point then there is a $j$ such that $\pi_{g_i}(\alpha_{g_j})=\alpha_{g_j}$. Then $g_i^{-1}g_j=g_j$ which implies that $g_i={\text{Id}}(G)$, i.e. $p=1$. This yields a contradiction and we finish the proof.  
 \end{proof}
 From the definition of $E_{\pi}(C(v))$ and $F_{\pi}(C(v))$ we quickly draw the following conclusion:
 \begin{theorem}\label{th 4.3}
 	Let $G$ be a finite group of order $n$ and $\pi$ be an automorphism of $C(v)$ of type $p-(c,0)$ for $v\in \mathbb{F}_2G$. Then the dimension of $F_{\pi}(C(v))$ is less than or equal to $c$, the minimum distance of $F_{\pi}(C(v))$ is a multiple of $p$ and the minimum distance of $E_{\pi}(C(v))$ is even. 
 \end{theorem}
 
 \section{Computational results}\label{Sec 5}
 In this section, we will search for binary linear codes with good parameters based on Sections \ref{Sec.3} and \ref{Sec.4}. All the upcoming computational results were obtained by performing searches in the software package Magma(\cite{bib6}). 
 \\ \indent  We use the computer to search for a suitable $v\in \mathbb{F}_2G_s$ such that $C(v)$ is a linear code with new parameters. However, if we use linear search, the number of times we search will reach $2^n$, where $n=\lvert G_s\rvert$. The number of searches explodes exponentially as $n$ increases, so this is obviously not advisable. Thus we use random searches. Notice that if we choose different $v_1,v_2\in \mathbb{F}_2G_s$, $C(v_1)$ may be equal to $C(v_2)$. 
  \begin{theorem}\label{th 5.1}
 	Let $G$ be a finite group of order $n$ and $v_1,v_2\in \mathbb{F}_qG$, if $v_2$ is invertible in $\mathbb{F}_qG$ then $C(v_1)=C(v_2v_1)$.
 \end{theorem}
 \begin{proof}
 	It is obvious that $C(v_2v_1)\subset C(v_1)$. Since $v_2$ is invertible, then there exists  $v_3\in \mathbb{F}_qG$ such that $v_3v_2=1$. We have $C(v_1)=C(v_3v_2v_1)\subset C(v_2v_1)$, thus $C(v_1)=C(v_2v_1)$.
 \end{proof}
 From Theorem \ref{th 5.1} we know that if a code $C(v)$ is with parameters $[n,k,d]$, there are many $v'\in \mathbb{F}_qG$ such that $C(v)=C(v')$. The number of $v'$ that satisfy this condition depends on the number of invertible elements in $\mathbb{F}_qG$. Therefore, it is feasible for us to adopt computer search. We list our findings in Table \ref{table 1} and the values of $v_i$ are listed in Table \ref{table2}. In Table \ref{table 1}, we write new record-breaking codes in bold. Further, we give the order of the permutation automorphism groups of these codes in the last column.
 \begin{table}[h]
 	\begin{minipage}{174pt}\caption{New record-breaking binary linear codes from Section \ref{Sec.3}}\label{table 1}
 		\begin{tabular}{@{}llllll@{}}
 			\toprule
 			$v$ & Type of $G$ & $C(v)$ & $C(v)^{\bot}$ & Best known codes in \cite{bib5}&$\lvert \text{PAut}(C(v))\rvert$ \\ 
 			\midrule
 			$v_1$&$G_1$&${\bf [54,31,10]}_2$&$[54,23,12]_2$&$[54,31,9]_2$&108\\ 
 			$v_2$&$G_2$&$ [78,54,6]_2$&$\textbf{[78,24,24]}_2$&$[78,24,22]_2$&312\\
 			$v_3$&$G_1$&$ \textbf{[81,48,12]}_2$&$[81,33,15]_2$&$[81,48,11]_2$&81\\ 
 			$v_4$&$G_1$&$ \textbf{[81,54,10]}_2$&$[81,27,18]_2$&$[81,54,9]_2$&81\\ 
 			$v_5$&$G_1$&$ \textbf{[84,51,12]}_2$&$[84,33,14]_2$&$[84,51,11]_2$&84\\ 
 			$v_6$&$G_2$&$ [90,64,8]_2$&$\textbf{[90,26,25]}_2$&$[90,26,24]_2$&90\\
 			$v_7$&$G_2$&$ [98,61,10]_2$&$\textbf{[98,37,22]}_2$&$[98,37,21]_2$&98\\
 			$v_8$&$G_2$&$ [98,60,12]_2$&$\textbf{[98,38,21]}_2$&$[98,38,20]_2$&98\\  
 			$v_9$&$G_2$&$ [100,72,7]_2$&$\textbf{[100,28,28]}_2$&$[100,28,27]_2$&100\\
 			$v_{10}$&$G_2$&$[100,64,10]_2$ &$\textbf{[100,36,24]}_2$&$[100,36,23]_2$&100\\
 			$v_{11}$&$G_2$&$[105,78,7]_2$ &$\textbf{[105,27,32]}_2$&$[105,27,30]_2$&105\\
 			$v_{12}$&$G_2$&$[105,72,10]_2$ &$\textbf{[105,33,27]}_2$&$[105,33,26]_2$&105\\
 			$v_{13}$&$G_2$&$[105,76,9]_2$ &$\textbf{[105,29,30]}_2$&$[105,29,28]_2$&105\\
 			$v_{14}$&$G_2$&$[105,73,8]_2$ &$\textbf{[105,32,28]}_2$&$[105,32,27]_2$&105\\
 			$v_{15}$&$G_2$&$[110,89,6]_2$ &$\textbf{[110,21,38]}_2$&$[110,21,36]_2$&110\\
 			\bottomrule
 		\end{tabular}
 	\end{minipage}
 \end{table}
 \begin{table}[h]
 	\begin{minipage}{174pt}\caption{The values of $v_i$}\label{table2}
 		\begin{tabular}{@{}lll@{}}
 			\toprule
 			
 			$v_i$ &Parameters of $G$& $\Psi(v_i)$ \\ 
 			\midrule
 			$v_1$&$(n,m,k)=(9,6,2)$&\begin{tabular}[c]{@{}l@{}}
 				$(0 1 1 1 1 1 1 0 1 0 0 0 0 0 0 1 0 0 0 0 1 0 11 1 0 0 1 0 0 1 0 1 0 1 0 0 1 0 11 0 0 1$\\ $ 0 1 1 1 1 11 0 1 1)$
 			\end{tabular}\\
 			$v_2$&$(n_1,k_1,n_2,k_2,m)=(2,1,13,3,3)$&\begin{tabular}[c]{@{}l@{}}
 				$(1 0 0 0 0 0 1 0 1 0 1 1 1 0 0 1 0 0 1 1 1 1 0 0
 				0 1 1 0 0 1 0 1 1 1 0 1 0 1 1 1 1 0 1 0 $\\ $1 0 1 1 0 0 0
 				0 0 1 1 1 0 1 1 1 1 0 0 0 0 0 1 0 1 1 0 1 1 0 1 0 0 0)$
 			\end{tabular}\\
 			$v_3$&$(n,m,k)=(9,9,4)$&\begin{tabular}[c]{@{}l@{}}
 				$(0 1 0 0 1 0 1 1 1 0 0 1 0 0 1 0 1 1 0 0 0 0 0 
 				1 0 1 1 1 1 1 0 1 1 1 1 0 0 1 0 0 1 0 1 0$\\ $ 1 1 1 0 1 0
 				1 1 0 0 1 1 0 1 1 0 1 0 0 0 1 1 1 1 1 1 1 1 0 0 1 1 1
 				1 1 0 0)$
 			\end{tabular}\\
 			$v_4$&$(n,m,k)=(9,9,4)$&\begin{tabular}[c]{@{}l@{}}
 				$(1 1 0 0 1 0 0 0 1 1 1 0 1 0 0 0 0 0 1 0 0 1 0 
 				0 1 0 0 1 0 1 1 0 0 1 0 1 1 0 0 1 0 0 1 1$\\ $ 1 1 0 1 1 0
 				0 1 1 0 0 0 0 1 1 0 0 0 0 0 0 0 1 0 1 0 0 1 0 0 1 0 1
 				0 1 0 1)$
 			\end{tabular}\\
 			$v_5$&$(n,m,k)=(7,12,4)$&\begin{tabular}[c]{@{}l@{}}
 				$(0 1 1 0 0 0 0 1 1 0 0 0 0 0 0 1 0 1 1 0 0 1 0 
 				0 1 1 0 1 0 0 1 1 1 1 0 1 0 1 1 1 1 0 0 0 $\\ $1 1 1 0 1 0
 				1 1 1 0 1 0 1 1 0 0 0 1 0 0 1 1 0 1 0 1 1 1 1 1 1 1 0
 				1 1 1 0 1 1 0)$
 			\end{tabular}\\
 			$v_6$&$(n_1,k_1,n_2,k_2,m)=(3,1,15,4,2)$&\begin{tabular}[c]{@{}l@{}}
 				$(0 1 0 1 0 1 0 0 0 1 0 1 0 0 0 0 1 1 0 1 1 0 1 1
 				1 0 0 1 0 0 0 1 0 1 1 1 0 0 0 0 0 1 0 1$\\ $ 1 0 1 0 0 0 0
 				0 0 1 0 0 0 0 1 0 1 1 1 1 0 0 0 0 1 1 1 0 0 0 0 1 0 1
 				1 0 0 0 1 1 0 0 1 0 1 $\\ $1)$
 			\end{tabular}\\
 			$v_7$&$(n_1,k_1,n_2,k_2,m)=(7,6,7,1,2)$&\begin{tabular}[c]{@{}l@{}}
 				$(1 0 1 0 1 0 1 1 1 0 1 0 0 1 1 1 1 0 0 0 0 1 1 1
 				1 1 0 1 0 1 0 0 0 1 0 0 0 1 0 0 0 0 0 0 $\\ $0 1 1 1 0 1 1
 				1 1 0 1 1 1 1 1 0 1 0 0 0 1 0 0 0 1 0 0 1 1 1 1 0 0 0
 				0 0 1 0 1 0 0 1 1 1 0$\\ $ 1 0 1 1 1 0 0 0 0)$
 			\end{tabular}\\
 			$v_8$&$(n_1,k_1,n_2,k_2,m)=(7,6,7,1,2)$&\begin{tabular}[c]{@{}l@{}}
 				$(0 0 0 0 0 0 0 0 0 1 1 0 1 0 0 1 0 1 0 1 0 1 1 
 				0 1 0 1 0 1 1 1 1 1 1 0 1 0 0 1 1 0 1 0 0 $\\ $0 1 1 0 0 1
 				1 0 1 1 0 1 0 1 1 0 1 0 1 0 1 0 0 0 1 0 0 0 0 1 1 0 1
 				1 0 0 1 0 0 0 0 1 0 1 1 $\\ $1 0 0 0 1 0 0 0 1)$
 			\end{tabular}\\
 			$v_9$&$(n_1,k_1,n_2,k_2,m)=(5,2,5,2,4)$&\begin{tabular}[c]{@{}l@{}}
 				$(1 0 1 0 0 0 0 0 1 1 0 1 0 0 0 0 1 0 1 0 0 0 1 0 0
 				0 0 0 0 1 1 1 1 0 0 1 0 1 1 1 1 0 0 0$\\ $ 0 0 0 1 0 0 1 1
 				1 1 0 0 0 0 1 0 0 0 0 0 1 1 1 0 0 1 1 0 0 1 0 0 1 1 0
 				1 0 0 1 1 0 0 0 1 0 $\\ $1 0 1 1 0 0 1 1 1 1 1)$
 			\end{tabular}\\
 			$v_{10}$&$(n_1,k_1,n_2,k_2,m)=(5,2,5,2,4)$&\begin{tabular}[c]{@{}l@{}}
 				$(1 0 1 1 1 0 0 0 1 0 0 0 0 0 0 1 0 1 0 0 1 1 0 
 				0 1 0 1 1 1 0 0 1 0 1 1 1 0 0 1 1 0 1 1 1 1$\\ $ 0 1 0 1 0
 				1 0 1 0 0 1 0 0 1 0 0 1 1 0 0 1 0 1 1 1 0 0 1 0 0 0 0
 				1 0 1 0 1 0 0 0 0 0 0 1 0 1$\\ $ 1 1 0 1 0 1 0 0 0)$
 			\end{tabular}\\
 			$v_{11}$&$(n_1,k_1,n_2,k_2,m)=(7,2,5,1,3)$&\begin{tabular}[c]{@{}l@{}}
 				$(1 1 1 0 0 1 0 0 1 0 1 1 1 1 1 1 1 1 0 1 1 0 0 1
 				1 1 1 1 1 0 1 1 0 1 0 0 0 0 1 0 0 0 1 0 $\\ $0 0 0 0 1 1 0
 				1 1 1 0 0 0 1 0 0 0 0 0 1 1 0 0 0 1 0 1 1 0 0 1 0 1 1
 				1 1 1 0 1 0 1 1 1 1 1$\\ $ 0 1 0 0 1 1 1 1 0 0 1 1 1 0 0 0)$
 			\end{tabular}\\
 			$v_{12}$&$(n_1,k_1,n_2,k_2,m)=(7,2,5,1,3)$&\begin{tabular}[c]{@{}l@{}}
 				$(0 0 0 1 0 0 1 1 0 1 0 1 0 1 1 1 1 0 0 1 0 0 1 
 				0 1 1 1 1 0 0 0 0 0 0 1 1 1 0 0 1 0 1 0 0$\\ $ 1 0 1 0 0 1
 				0 0 1 0 1 1 0 1 1 1 0 0 1 1 1 0 1 0 0 1 1 0 1 0 1 1 1
 				1 0 0 0 1 1 1 0 0 0 0 1$\\ $ 0 1 1 0 0 0 0 1 0 0 0 1 0 1 1
 				0 )$
 			\end{tabular}\\
 			$v_{13}$&$(n_1,k_1,n_2,k_2,m)=(7,2,5,1,3)$&\begin{tabular}[c]{@{}l@{}}
 				$(1 0 0 0 1 0 1 0 0 0 1 0 0 1 1 0 1 0 0 0 0 1 1 1
 				1 0 0 0 0 1 1 1 0 0 1 0 0 1 1 1 1 1 1 0 $\\ $0 1 0 0 0 0 0
 				1 0 0 1 1 0 0 1 1 0 1 1 1 1 0 0 1 1 1 0 1 1 0 1 1 1 0
 				1 0 0 0 1 1 0 1 1 1 1 $\\ $1 0 1 1 1 1 1 0 0 0 1 0 0 1 1 0 )$
 			\end{tabular}\\
 			$v_{14}$&$(n_1,k_1,n_2,k_2,m)=(7,2,5,1,3)$&\begin{tabular}[c]{@{}l@{}}
 				$(1 0 1 1 0 1 1 0 0 1 1 1 1 1 0 1 1 1 1 0 1 0 0 
 				0 1 1 0 0 0 0 1 1 0 1 1 1 0 1 1 0 1 0 1 1$\\ $ 0 0 1 0 1 0
 				0 1 0 0 1 1 1 0 1 0 0 0 1 1 0 0 1 1 1 0 0 1 0 0 1 0 0
 				1 0 0 0 1 1 1 0 1 0 1 1 $\\ $0 0 1 1 1 0 1 1 1 1 1 0 0 1 1
 				0 )$
 			\end{tabular}\\
 			$v_{15}$&$(n_1,k_1,n_2,k_2,m)=(2,1,11,4,5)$&\begin{tabular}[c]{@{}l@{}}
 				$(1 0 0 1 1 1 1 0 1 0 0 1 1 0 0 0 1 0 0 1 0 1 0 1 
 				1 1 0 1 0 0 1 0 0 0 0 1 1 1 1 0 0 1 1 1$\\ $ 1 0 1 1 1 1 0
 				0 0 1 0 0 0 1 0 1 0 1 0 0 0 1 1 1 0 1 1 1 0 1 1 0 1 1
 				1 0 1 0 1 1 0 0 1 1 1$\\ $ 1 0 0 1 0 0 0 1 1 0 1 0 1 0 0 1
 				0 0 0 1 1 )$
 			\end{tabular}\\
 			\bottomrule
 		\end{tabular}
 	\end{minipage}
 \end{table}
 \\ \indent From Theorem $\ref{th 4.3}$ we know that if there exists a binary $[n,k,d]$ linear code $\mathcal{C}$ with an automorphism $\pi$ of type $p$-$(c,0)$, where $d$ is odd. Then the minimum distance of $F_{\pi}(\mathcal{C})$ is greater than $d$. Using this result we can construct the following linear code which has a larger minimum distance than the best known code $[108,23,35]_2$ in \cite{bib5}.
 \begin{theorem}\label{Th 5.2}
 	There is a binary $\textbf{[108,23,36]}$ linear code. 
 \end{theorem}
 \begin{proof}
 	Let $G_2=(C_7\times C_5)\rtimes C_3=\left\langle x,y,z\mid x^7=y^5=z^3=1,x^z=x^2,y^z=y \right\rangle$ and $v\in \mathbb{F}_2G_2$, where $\Psi(v)=(1 0 0 0 1 0 0 1 0 0 1 1 0 1 0 1 1 1 1 0 1 1 0 1 0 0 0 1 1 1 1 0 0 1 1 0 0 0 1 1 1 0 1 0 1 1 0 0 1 1 1 0 0 0 \\0 0 1 1 1 1 0 1 0 1 0 1 1 0 0 0 1 1 0 0 1 1 0 0 0 0 0 1 1 1 0 1 1 1 0 0 0 1 1 1 0 0 1 1 0 1 1 1 1 1 0)$. According to Magma's calculations, $C(v)^{\bot}$ is a binary $[105,23,33]$ linear code and its automorphism group is iosmorphic to $G_2$. We select any element $\pi$ of order 35 in $\text{PAut}(C(v)^{\bot})$, then $E_{\pi}(C(v)^{\bot})$ is with parameters $[105,22,36]$. From Theorem \ref{th 4.1}, $E_{\pi}(C(v)^{\bot})$ is a subcode of $C(v)^{\bot}$. Let $\mathcal{C}'$ be the best known binary $[3,1,3]$ linear code. Then by using Construction X for $C(v)^{\bot}$, $E_{\pi}(C(v)^{\bot})$ and $ \mathcal{C}'$, we can obtain a binary linear code $\mathcal{C}''$ with parameters $[108,23,36]$. 
	Moreover, $\lvert \text{PAut}(\mathcal{C}'')\rvert=630$. 
 \end{proof}
 However, when the minimum distance of a binary linear code is even, the minimum distance of its subcode may also be larger than that of the original code. We have decomposed all the codes in Table \ref{table 1} and we obtain the following new code:
 \begin{theorem}\label{Th 5.3}
 	There is a binary $\textbf{[108,29,32]}$ linear code.
 \end{theorem}
 \begin{proof}
 	Let $\pi$ be an any element of order 35 in $\text{PAut}(C(v_{13})^{\bot})$. Then $E_{\pi}(C(v_{13})^{\bot})$ is with parameters $[105,27,32]$. From Theorem \ref{th 4.1}, $E_{\pi}(C(v_{13})^{\bot})$ is a subcode of $C(v_{13})^{\bot}$. Let $\mathcal{C}'$ be the best known binary $[3,2,2]$ linear code, then by using Construction X for $C(v_{13})^{\bot}$, $E_{\pi}(C(v_{13})^{\bot}$ and $ \mathcal{C}'$, we can obtain a binary linear code $\mathcal{C}''$ with parameters $[108,29,32]$. Moreover, $\lvert \text{PAut}(\mathcal{C}'')\rvert=105$. 
 \end{proof}
 In fact, some special groups can  also be used to construct new record-breaking linear codes. Magma includes all groups of order up to 2000 excluding the groups of order 1024. We choose groups SmallGroup(108,9) and SmallGroup(120,5), then binary codes $\textbf{[108,16,42]}$ and $\textbf{[120,31,34]}$ can be obtained from above two groups. All of the generator matrices and related data in Theorem \ref{Th 5.2} and \ref{Th 5.3} are showed online at \url{https://github.com/Yucong1234/New_binary_linear_codes}. 
 \begin{theorem}\label{The 5.4}(\cite{bib3})
 	Let $\mathcal{C}$ be a binary $[n,k,d]$ code. Then we can construct an $[n-1,k,d-1]$ code by puncturing, an $[n-1,k-1,d]$ code by shortening and an $[n+1,k,d']$ code by extending, where $d'=d+1$ if $d$ is odd else $d'=d$. 
 \end{theorem}
 Then from Theorem \ref{Th 5.2}-\ref{The 5.4}, Table \ref{table 1} and the above two record-breaking codes, we can obtain more new codes. We list the results in Table \ref{table 3}.
 Notice that $P^iS^jE^k$ indicates that a code is punctured by $i$ times, shortened by $j$ times and extended by $k$ times.   
  \begin{table}[h]
 	\begin{minipage}{174pt}\caption{New record-breaking binary linear codes and the corresponding construction methods.}\label{table 3}
 		\begin{tabular}{@{}ll@{}}
 			\toprule
 			Parameters&Construction methods\\ 
 			\midrule
 			$[54-i-j,31-i,10-j], 0\le i \le 3, 0\le j \le 1$&$P^jS^i$ $[54,31,10]$\\
 			$[78-i-j,24-i,24-j], 0\le i \le 2, 0\le j \le 1$&$P^jS^i$ $[78,24,24]$\\
 			$[78-j,24,24-j],  2\le j \le 3$&$P^j$ $[78,24,24]$\\
 			$[79,24,24]$&$E^1$ $[78,24,24]$\\
 			$[81-i-j,48-i,12-j], 0\le i \le 1, 0\le j \le 1$&$P^jS^i$ $[81,48,12]$\\
 			$[81-i-j,54-i,10-j], 0\le i \le 2, 0\le j \le 1$&$P^jS^i$ $[81,54,10]$\\
 			$[84-i-j,51-i,12-j], 0\le i \le 2, 0\le j \le 1$&$P^jS^i$ $[84,51,12]$\\
 			$[91,26,26]$&$E^1$ $[90,26,25]$\\
 			$[98-i-j,37-i,22-j], 0\le i \le 1, 0\le j \le 1$&$P^jS^i$ $[98,37,22]$\\
 			$[99,38,22]$&$E^1$ $[98,38,21]$\\
 			$[99,28,27]$&$P^1$ $[100,28,28]$\\
 			$[99,36,23]$&$P^1$ $[100,36,24]$\\
 			$[104,27,31]$&$P^1$ $[105,27,32]$\\
 			$[106,27,32]$&$E^1$ $[105,27,32]$\\
 			$[105-i-j,29-i,30-j], 0\le i \le 1, 0\le j \le 1$&$P^jS^i$ $[105,29,30]$\\
 			$[106,29,30]$&$E^1$ $[105,29,30]$\\
 			$[105-i-j,32-i,28-j], 0\le i \le 2, 0\le j \le 1$&$P^jS^i$ $[105,32,28]$\\
 			$[105+i,33,28], 1\le i \le 2$&$E^i$ $[105,32,28]$\\
 			$[107,16,41]$&$P^1$ $[108,16,42]$\\
 			$[107,23,35]$&$P^1$ $[108,23,36]$\\
 			$[108-i-j,29-i,32-j], 0\le i \le 2, 0\le j \le 1$&$P^jS^i$ $[108,29,32]$\\
 			$[109,21,37]$&$P^1$ $[110,21,38]$\\
 			$[111,21,38]$&$P^1$ $[110,21,38]$\\
 			$[119,31,33]$&$P^1$ $[120,31,34]$\\
 			\bottomrule
 		\end{tabular}
 	\end{minipage}
 \end{table}

 \section{Conclusion}\label{Sec 6}
 In this work, we employ $G$-matrix to find new record-breaking binary linear codes, where $G$ is the semidirect product of an abelian group and a cyclic group. By puncturing, shortening and extending, we obtain a total of  77 new record-breaking binary linear codes. In this paper, we only consider the binary case. In fact, the $G$-matrix can generate many  record-breaking non-binary linear codes. In addition, due to the limitation of computing power, we only search for linear codes with lengths less than or equal to 120. Interested readers can try non-binary cases or larger code lengths.
 \section*{Data availability}
 The authors have no conflict of interest related to the content of this article to declare.
 \section*{Acknowledgments}
 We would like to express our sincere appreciation to all those who have contributed to this paper. Special thanks to Prof. Markus Grassl for his help with our paper. The code in Theorem \ref{Th 5.3} was found with his help. This work was supported by the National Natural Science Foundation of China under Grant Nos U21A20428 and 12171134.
 \section*{Declaration}
 For now, this is only the initial version which was completed in April 2024, and the authors will add some new results in the future. If a reader has done something similar to us, to avoid duplication, please cite or contact us appropriately.

\bibliographystyle{unsrt}  


\end{document}